\documentclass[11pt]{llncs}
\usepackage[margin=1in]{geometry}
\usepackage{graphicx}

\usepackage[english]{babel}
\usepackage{amsfonts}
\usepackage{inputenc}
\usepackage{color}
\usepackage{algpseudocode}
\usepackage{booktabs,multicol,multirow}
\usepackage{fancyhdr}
\usepackage{epic}
\usepackage{xspace}
\usepackage{amssymb,amsmath,verbatim}
\usepackage{multirow}
\usepackage{multicol}
\usepackage{adjustbox}
\usepackage{threeparttable}
\usepackage{float}
\usepackage[dvipsnames]{xcolor}

\usepackage{amsthm}
\usepackage{chngcntr}

\usepackage[linesnumbered,ruled,vlined]{algorithm2e}
\SetKwInput{KwInput}{Input}                
\SetKwInput{KwOutput}{Output}              

\newtheorem{thm}{Theorem}

\newtheorem{lem}[thm]{Lemma}

\newtheorem{defn}[thm]{Definition}

\usepackage{url}

\def\ZZ{\mathbb{Z}}

\def\RR{\mathbb{R}}

\def\cal{\mathcal}
\def\bf{\mathbf}

\def\MPK{\mathsf{MPK}}
\def\MSK{\mathsf{MSK}}
\def\FRD{\mathsf{FRD}}

\def\SK{\mathsf{SK}}
\def\Tsk{\mathsf{Tsk}}
\def\CT{\mathsf{CT}}
\def\id{\mathsf{id}}

\def\SampleRight{\mathsf{SampleRight}}
\def\SamplePre{\mathsf{SamplePre}}
\def\SampleBasisLeft{\mathsf{SampleBasisLeft}}
\def\SampleBasisRight{\mathsf{SampleBasisRight}}

\def\ExtendBasis{\mathsf{ExtendBasis}}

\def\Decrypt{\mathsf{Decrypt}}
\def\Setup{\mathsf{Setup}}

\def\Derive{\mathsf{Derive}}
\def\Extract{\mathsf{Extract}}
\def\TskGen{\mathsf{TskGen}}
\def\TkVer{\mathsf{TkVer}}
\def\Encrypt{\mathsf{Encrypt}}
\def\Decrypt{\mathsf{Decrypt}}
\def\Round{\mathsf{Round}}
\def\ReRand{\mathsf{ReRand}}
\def\Invert{\mathsf{Invert}}
\def\Tsk{\mathsf{Tsk}}
\def\idr{\mathrm{Pr}} 
\def\Adv{\mathsf{Adv}}
\def\Anon{\textsf{ANON-sID-CPA}}
\def\IND{\textsf{IND-sID-CPA}}

\def\u{\bf{u}}
\def\A{\bf{A}}
\def\R{\bf{R}}
\def\B{\bf{B}}
\def\F{\bf{F}}
\def\D{\bf{D}}
\def\T{\bf{T}}
\def\G{\bf{G}}
\def\C{\bf{C}}
\def\U{\bf{U}}
\def\v{\bf{v}}
\def\z{\bf{z}}

\def\e{\bf{e}}

\def\m{\bf{m}}
\def\x{\bf{x}}
\def\k{\bf{k}}
\def\y{\bf{y}}

\def\negl{\mathsf{negl}}

\def\L{\Lambda}
\def\Lp{\Lambda^{\perp}}
\def\b{\bf{b}}
\def\s{\bf{s}}
\def\c{\bf{c}}

\def\LWE{\textsf{LWE}}
\def\SIVP{\textsf{SIVP}}
\def\GapSVP{\textsf{GapSVP}}

\begin{document}
	\title{Anonymous communication system provides a secure environment without leaking metadata, which has many application scenarios in IoT
	}
	\titlerunning{}
	\authorrunning{}
	\author{Ngoc Ai Van Nguyen\inst{1}\and  Minh Thuy Truc Pham\inst{2}}
	\institute{
		Department of Mathematics and Physics, University of Information Technology,\\
		 Vietnam National University, Ho Chi Minh city, Vietnam\\
		 \email{vannna@uit.edu.vn}
		\and 
	 Institute of Cybersecurity and Cryptology\\
		School of Computing and Information Technology, University of Wollongong\\
		Northfields Avenue, Wollongong NSW 2522, Australia\\	
		\email{pm.thuytruc@gmail.com}
	}

	\maketitle              
	
	\begin{abstract}
    Anonymous Identity Based Encryption (AIBET) scheme allows a tracer to use the tracing key to reveal the recipient's identity from the ciphertext while keeping other data anonymous. This special feature makes AIBET a promising solution to distributed IoT data security. In this paper, we construct an efficient quantum-safe Hierarchical Identity-Based cryptosystem with Traceable Identities (AHIBET) with fully anonymous ciphertexts. We prove the security of the AHIBET scheme under the Learning with Errors (LWE) problem in the standard model. 
	\end{abstract}

\section{Introduction}
Internet of Things (IoT) has emerged as a set of interconnected technologies like Wireless Sensors Networks (WSN) and Radio Frequency Identification (RFID), that provide identification, computation, and mutual information exchange among the connected devices all over the world. The key idea of the IoT is to obtain information about our environment to understand and control and act on it~\cite{DMR16}.

Identity-Based Encryption (IBE) is a type of public-key encryption where the public key is an arbitrary string that uniquely defined the user (for example an email address or a telephone number). The Private-key Generator (PKG) who has knowledge of a master secret key generates the private key for the corresponding identities. This concept was first introduced by Shamir~\cite{Shamir} and then in 2001, Boneh and Franklin~\cite{BF03} proposed the first practical IBE scheme based on bilinear map. 
The idea of Hierarchical Identity-Based Encryption (HIBE), which is an extension of IBE where each level can issue private keys for identites of the next level, was first proposed in the work of Gentry and Silverberg~\cite{GS02}. Since then, there have been many efficient constructions of HIBE, ranging from classical setting~\cite{BonehB04a,BonehBG05,Waters05} to post-quantum setting~\cite{ABB10-EuroCrypt,CHKP10,SRB14} just to name a few. 

The concept of ``Anonymous" IBE offers an additional privacy guarantee to standard IBE schemes where the ciphertexts do not leak the identity of the recipients. AIBE is a promising solution to anonymous communications and it can be applied to many realistic scenarios that provide privacy-preserving and security under cloud environment. It can also bring a secure environment without leaking metadata which has many application scenarios in the aforementioned distributed IoT system~\cite{Jiang18}.
However, the first AIBE construction of Boneh and Frankl in~\cite{BF03} is just anonymous in the random oracle model and it was a challenging problem to achieve anonymous IBE in the standard model until~\cite{BW06}. In~\cite{BW06HIBE}, Boyen and Waters proposed the first secure anonymous HIBE scheme without random oracles. More recently, the HIBE constructions in the post-quantum setting~\cite{ABB10-EuroCrypt,SRB14} are proven to be anonymous secure in the standard model in the mean of a ciphertext encrypted for a target identity is indistinguishable from a random element in the ciphertext space which helps hide this identity from any malicious attacker. 

Although this strong unconditional privacy seems very attractive from the user's point of view, it can potentially be a dangerous tool against public safety if there is no way to revoke such privacy when illegal behavior is detected. For example, in the case where the email filtering system has to filter out all encrypted email from members are suspected of illegal activity, standard anonymous IBE and HIBE prevent the system reveal the recipients of those ciphertexts. Traceability can provide a solution to this problem in which an additional traceability function can detect specific identities in ciphertexts and all the others remain anonymous. 

In 2019, Blazy et al.~\cite{Blazy19} first considered the traceability for identity-based encryption and constructed an Anonymous Identity Based Encryption (AIBET) scheme in the standard model but under the matrix Diffie Hellman (MDDH) assumption. Two security notions are formally defined in [4] are \textit{anonymity} and \textit{ciphertext indistinguishability}. Anonymity requires that someone without an associated user secret key or tracing key should not be able to guess the targeted identity. The notion of indistinguishability requires that no one can distinguish between a valid ciphertext and a random string from the ciphertext space even having access to the tracing key of the target identity. Recent, in~\cite{LiuTTMC21}, Liu et al. proposed a lattice-based construction for AIBET which is based on the anonymous IBE by Katsumata and Yamada~\cite{KY16}. However, they do not address the notion of indistinguishability which is the main difference between an AIBET and a standard anonymous IBE. Note that the role of the tracer and the Private-key Generator PKG are distinguishable where the tracer has less power than the PKG. For example, it could be a gateway that checks whether an email for a suspected illegal user is passed. Even if the tracers are corrupted, the privacy and the confidentiality of the system  will still remain intact.

\medskip
\noindent\textbf{Our contribution:} We propose a concrete construction of an Anonymous Lattice Hierarchical Identity-Based Encryption with Traceable Identities (AHIBET) scheme which is secure in the standard model based on the hardness assumption of lattices. In particular, our AHIBET construction is anonymous across all the levels of hierarchy, i.e., ciphertexts conceal recipients' identities from everyone which does not know the corresponding keys for decryption or tracing. Traceability cannot be extended down the hierarchy, i.e., knowing the tracing key for identity $\id$ doesn't imply knowing tracing key for any of its descendants. Besides, our construction is ciphertext indistinguishable, i.e., even having the tracing key for identity $\id$, one cannot distinguish the ciphertexts of message $\m$ from the one of random messages for identity $\id$. 

An instance of our AHIBET yields a lattice-based construction of AIBET that can be easily converted to a construction over ideal lattices using the techniques in~\cite{BertFRS18}, which outperforms the AIBET by Liu et al.~\cite{LiuTTMC21}\footnote{In fact, the public parameter in Liu et al.~\cite{LiuTTMC21} will be a factor of $dl$ greater than ours where $d$ is some fixed constant (e.g., d=2 or 3) and $l=\lceil n^{1/d}\rceil$ for $n$ the security parameter.}.

\medskip
\noindent\textbf{Technical Overview:}\\
The first main idea is that an AHIBET system must be controlled by three levels of trapdoors:
\begin{itemize}
    \item The master secret key $\MPK$ can be used to generate secret key $\SK_{\id}$ and tracing key $\Tsk_{\id}$ for each identity $\id$ of any level.
    
    \item The secret keys $\SK_{\id}$ enable recipients to decrypt the corresponding ciphertexts. Each secret key $\SK_{\id}$ can be used to generate the secret keys for identities of the next level and thus control all descendants of $\id$.
    
    \item The tracing keys $\Tsk_{\id}$ enable tracers to detect only the ciphertexts sent to identities $\id$ without leaking information of the messages. 
 \end{itemize}
 To achieve the identity traceability property, we attach each ciphertext a random tag and its encapsulation whereas tracing keys are the trapdoors for decapsulation.

We exploit the power of lattice trapdoors in \cite{MP12,CHKP10} combining with the HIBE construction by Agrawal et al.~\cite{ABB10-EuroCrypt} to achieve our AHIBET.

In \cite{ABB10-EuroCrypt}, each identity $\id$ is assigned a matrix $\F_{\id}$ and message $\m$ is encrypted following the dual-Regev scheme:
    $$ \c^T = \s^T \F_{\id} + \e^T ,  \ \ \ \c'^T = \s^T \U + \e'^T + \m^T \left\lfloor \displaystyle{\frac{q}{2}} \right\rfloor.$$
    
     In our scheme, we use one dual-Regev part to encrypt the message and another one to encapsulate the random tag to allow the ciphertext to reveal the recipients' identity from the tracing key holder. 

In \cite{MP12}, the authors introduced a so-called $\G$-trapdoor where $\G$ is a gadget matrix in $\ZZ_q^{n \times \omega}$.  A $\G$-trapdoor for matrix $\F$ is a matrix $\R\in\ZZ^{(m-\omega)\times \omega}$ such that $\F=[\A|\A\R+\bf{H}\G]$ for some invertible matrix $\bf{H}\in \ZZ^{n\times n}_q$ where $\A\in\ZZ^{n\times (m-\omega)}_q$.  The authors called it ``strong trapdoor" since a good basis $\T_{\F}\in \ZZ^{m \times m}$ for $\Lambda_q^{\perp}(\F)$ can be obtained from the knowledge of the matrix $\R$ but the reverse is hard. Moreover, with either $\R$ or $\T_{\F}$, one can easily generate a low norm matrix $\D_{\F}\in\ZZ^{n\times t}_q$ satisfying $\F.\D_{\F} = \U$ with respect to a given random matrix $\U\in\ZZ_q^{n\times t}$ using the sampling algorithms from~\cite{ABB10-EuroCrypt} and~\cite{MP12}. Since $\D_{\F}$ is a kind of weaker trapdoor than $\T_{\F}$, we can use such matrices $\R$, $\T_{\F}$ and $\D_{\F}$  as the three levels of trapdoors $\MSK$,  $\SK_{\id}$,  $\Tsk_{\id}$ respectively for a traceable  identity-based encryption where the matrix $\F$ is associated to an identity $\id=(\id_1,\ldots,\id_{\ell})$, namely, $\F= \F_{\id} = [\A|\A_1 + \FRD(\id_1)\G| \ldots| \A_{\ell}+\FRD(\id_{\ell})\G]$ for the public matrices $\A, \A_1,\ldots,\A_{\ell}$ and the full-rank difference encoding function $\FRD$. However, such trapdoors do not guarantee the anonymity and even the secrecy of messages across the hierarchy of identities. For example, knowing $\D_{\F_{\id_1}}$ and $\D_{\F_{[\id_1|\id_2]}}$, one can easily find a low norm matrix $\T$ of the same size as $\D_{\F_{\id_1|\id_2}}$ such that  $\F_{[\id_1|\id_2]} \T = \bf{0}$, which reveals information of the messages.
Therefore, we use a collision resistance hash function $\mathsf{H}$ to construct a matrix $\F'_{\id}=[\A|\A_0+\FRD(\mathsf{H}(\id))\G]$ and use the sampling algorithms to generate the tracing key $\D_{\F'_{\id}}$ of the identity $\id$ such that $\F'_{\id}\D_{\F'_{\id}}=\U$. Such tracing keys  are determined uniquely by the identities and independent of the secret keys, which ensures the anonymity and secrecy of the messages.



\section{Preliminaries}
\subsection{Anonymous Lattice Hierarchical Identity-Based Encryption with Traceable Identities (AHIBET)}
In this section, we describe the model of Anonymous Lattice Hierarchical Identity-Based Encryption with Traceable Identities (AHIBET) based on the Anonymous Lattice Identity-Based Encryption with Traceable Identities (AIBET) from~\cite{Blazy19} and its security model.

\begin{defn}[AHIBET] \label{def:AHIBET}
	An AHIBET scheme consists of the following seven algorithms:
	\begin{itemize}
		\item $\Setup(\lambda,d)$ takes as input the security parameter $\lambda$ and the maximal hierarchy depth $d$ of the scheme and outputs the master public key $\MPK$ and the master secret key $\MSK$.
		\item $\Extract(\MPK,\MSK,\id)$ uses the master public key $\MPK$ and the master secret key $\MSK$ to generate the secret key $\SK_{\id}$ for an identity $\id$ at depth $1$.  	
		\item $\Derive(\MPK,\SK_{\id}),(\id|\id_{\ell})$ takes as input the master public key $\MPK$ and a secret key $\SK_{\id}$ corresponding to an identity $\id$ at depth $\ell-1$, outputs the secret key $\SK_{\id|\id_{\ell}}$ for the identity $(\id|\id_{\ell})$ at depth $\ell$.
		\item $\TskGen(\MPK,\MSK,\id)$ uses the master public key $\MPK$ and the master secret key $\MSK$ to generate the tracing key $\Tsk_{\id}$ for a given identity $\id$.
		\item $\Encrypt(\MPK,\id,\m)$ takes as input the master public key $\MPK$, a given identity $\id$ and a message $\m$, outputs the ciphertext $\CT$.
		\item $\Decrypt(\MPK,\CT,\SK_{\id})$ takes as input the master public key $\MPK$, a ciphertext $\CT$ and a secret key $\SK_{\id}$. The algorithm outputs the message $\m$ if $\CT$ is encrypted for $\id$; otherwise, it outputs the rejection symbol $\perp$.
		\item $\TkVer(\MPK,\id,\Tsk_{\id},\CT)$ takes as input the master public key $\MPK$, an identity $\id$ and a ciphertext $\CT$, uses the tracing key $\Tsk_{\id}$ to check whether a ciphertext $\CT$ is encrypted for $\id$. $\TkVer$ outputs $1$ if $\CT$ is for the user with identity $\id$; otherwise, it outputs $0$.
	\end{itemize}
\end{defn}

\noindent\textbf{Correctness and soundness.}\\
The \textit{correctness} of AHIBET scheme requires that if for all key pairs $(\MPK,\MSK)$ output by $\Setup$, all $1\le \ell\le d$, all identities $\id=(\id_1,\dots,\id_{\ell})$ where $\id_i\in \ZZ^{n}_q\setminus\{\bf{0}\}$ and all messages $\m\in\{0,1\}^{\lambda}$, it holds that
	$$\Pr\left[ {\begin{gathered}
		\Decrypt(\MPK,\SK_{\id},\CT)=\m \end{gathered}  
		\left| \begin{gathered}
		(\MPK,\MSK)\gets\Setup(\lambda,d)\\
		\SK_{\id}\gets\Derive(\MPK,\MSK,\id)\\
		\CT\gets\Encrypt(\MPK,\id,\m)\\
		1\gets \TkVer(\MPK,\id,\Tsk_{\id},\CT)
		\end{gathered}  \right.} \right]=1$$
and the \textit{soundness} of AHIBET requires 
	$$\Pr\left[ {\begin{gathered}
		\Decrypt(\MPK,\SK_{\id},\CT)=\perp \end{gathered}  
		\left| \begin{gathered}
		(\MPK,\MSK)\gets\Setup(\lambda,d)\\
		\SK_{\id}\gets\Derive(\MPK,\MSK,\id)\\
		\CT\gets\Encrypt(\MPK,\id,\m)\\
		0\gets \TkVer(\MPK,\id,\Tsk_{\id},\CT)
		\end{gathered}  \right.} \right]=1$$

\noindent\textbf{Security models of AHIBET.} 
For the security models, we give the definition of anonymity and ciphertext indistinguishability for the AHIBET scheme.  
\begin{itemize}
    \item \textbf{Anonymity} is the property that the adversary can not distinguish the encryption of a chosen message for a first chosen identity from the encryption on the same message for a second chosen identity. Similarly, the adversary can not decide whether a ciphertext it received from the challenger  was encrypted for a chosen challenge identity, or for a random identity in the identity space. The anonymity game, denoted $\Anon$, is played between an adversary $\mathcal{A}$ and a challenger $\cal{C}$,  provided that the adversary $\cal{A}$ does not have the corresponding tracing key of the challenge identity, is defined through the following game:
\begin{itemize}
	\item \textbf{Init:} The adversary $\cal{A}$ is given the maximum depth of the hierarchy $d$ and then $\cal{A}$ decides a target pattern $\id^*=(\id_1^*,\dots,\id_{\ell}^*)$, $\ell\le d$. 
	\item \textbf{Setup:} At the beginning of the game, the challenger $\cal{C}$ runs $\Setup(\lambda,d)$ to obtain $(\MPK,\MSK)$ and gives the resulting master public key $\MPK$ to the adversary $\cal{A}$.
	\item  \textbf{Phase 1:} $\cal{A}$ may adaptively make queries polynomial many times to the key derivation oracle $\cal{O}_{\Derive}$ and the tracing key oracle $\cal{O}_{\TskGen}$ where:
	\begin{itemize}
	    \item Oracle $\cal{O}_{\Derive}(\id)$ takes input an identity $\id$ different from $\id^{\ast}$ and its prefixes, returns the output of $\Derive(\MPK,\MSK,\id)$.
	    \item Oracle $\cal{O}_{\TskGen}(\id)$ takes input an identity $\id$ different from $\id^{\ast}$, returns the output of $\TskGen(\MPK,\MSK,\id)$.
	\end{itemize}

	\item \textbf{Challenge:}
	The adversary $\cal{A}$ chooses a message $\m\in\{0,1\}^{\lambda}$ and gives it to the challenger $\cal{C}$. $\cal{C}$ then selects a random bit $b\in\{0,1\}$ and a random identity $\id'$ in the identity space which has the same depth with the challenge identity $\id^{\ast}$. If $b = 0$, $\cal{C}$ runs $\CT_0^{\ast}\gets \Encrypt(\MPK,\id^{\ast},\m)$; otherwise, it runs $\CT_1^{\ast}\gets \Encrypt(\MPK,\id',\m)$. Finally, $\cal{C}$ passes $(\MPK,\CT_{b}^*)$ through to the adversary $\cal{A}$.
	\item \textbf{Phase 2:} $\cal{A}$ continues to issue additional key derivation and tracing key queries and $\cal{C}$ responds as in \textbf{Phase 1}.
	\item \textbf{Guess:} $\cal{A}$ outputs its guess $b'\in \{0,1\}$ and wins if $b' = b$.
\end{itemize}

The advantage of $\cal{A}$ is defined as
$$\Adv_{\cal{A},\text{AHIBET}}^{\Anon}:=\left|\idr[b=b']-\frac{1}{2}\right|.$$

    \item In the \textbf{ciphertext indistinguishability game}, we use a privacy property called \textit{indistinguishable from random} which means that 
    the challenge ciphertext encrypted for a given message $\m^{\ast}$ is computationally indistinguishable from a the challenge ciphertext encrypted for a random message $\m$ on the same challenge identity $\id^{\ast}$, even the adversary $\cal{A}$ has the corresponding tracing key $\Tsk_{\id^{\ast}}$ of $\id^{\ast}$. The $\IND$ security model is defined through the following game, played between an adversary $\mathcal{A}$ and a challenger $\cal{C}$:
    \begin{itemize}
	\item \textbf{Init:} The adversary $\cal{A}$ is given the maximum depth of the hierarchy $d$ and then $\cal{A}$ decides a target pattern $\id^*=(\id_1^*,\dots,\id_l^*)$, $\ell\le d$. 
	\item \textbf{Setup:} At the beginning of the game, the challenger $\cal{C}$ runs $\Setup(\lambda,d)$ to obtain $(\MPK,\MSK)$ and gives the resulting master public key $\MPK$ to the adversary $\cal{A}$.
	\item  \textbf{Phase 1:} $\cal{A}$ may adaptively make queries polynomial many times to the key derivation oracle $\cal{O}_{\Derive}$ and the tracing key oracle $\cal{O}_{\TskGen}$ where:
	\begin{itemize}
	    \item Oracle $\cal{O}_{\Derive}(\id)$ takes input an identity  $\id$ different from $\id^{\ast}$ and its prefixes, returns the output of $\Derive(\MPK,\MSK,\id)$.
	    \item Oracle $\cal{O}_{\TskGen}(\id)$ takes input an identity  $\id$ different from $\id^{\ast}$, returns the output of $\TskGen(\MPK,\MSK,\id)$.
	\end{itemize}

	\item \textbf{Challenge:}
	The adversary $\cal{A}$ chooses a message $\m^{\ast}\in\{0,1\}^{\lambda}$ and gives it to the challenger $\cal{C}$. $\cal{C}$ sets $\m_0=\m^{\ast}$ and chooses a random message $\m_1$ in the message space. $\cal{C}$ then selects a random bit $b\in\{0,1\}$. If $b = 0$, $\cal{C}$ runs $\CT_0^{\ast}\gets \Encrypt(\MPK,\id^{\ast},\m_0)$; otherwise, it runs $\CT_1^{\ast}\gets \Encrypt(\MPK,\id^{\ast},\m_1)$. Finally, $\cal{C}$ passes $(\MPK,\CT_b^*)$ through to the adversary $\cal{A}$.
	\item \textbf{Phase 2:} $\cal{A}$ continues to issue additional key derivation and tracing key queries and $\cal{C}$ responds as in \textbf{Phase 1}.
	\item \textbf{Guess:} $\cal{A}$ outputs its guess $b'\in \{0,1\}$ and wins if $b' = b$.
    \end{itemize}

The advantage of $\cal{A}$ is defined as
$$\Adv_{\cal{A},\text{AHIBET}}^{\IND}:=\left|\idr[b=b']-\frac{1}{2}\right|.$$

\end{itemize}


\subsection{Lattices}
A lattice $\Lambda$ in $\ZZ^m$ is a set of all integer linear combinations of (linearly independent) basis vectors $\bf{B}=\{\b_1,\cdots,\b_n\}\in\ZZ^m$, i.e.,
$$\Lambda:=\left\{\sum_{i=1}^n\b_ix_i | x_i\in\ZZ~\forall i=1,\cdots,n \right\}\subseteq\ZZ^m.$$
We call $n$ the rank of $\L$ and if $n=m$ we say that $\L$ is a full rank lattice. In this paper, we mainly consider full rank lattices containing $q\ZZ^m$, called $q$-ary lattices,
\begin{align*}
\L_q(\A) &:= \left\{ \e\in\ZZ^m ~\rm{s.t.}~ \exists \bf{s}\in\ZZ_q^n~\rm{where}~\A^T\bf{s}=\bf{e}\mod q \right\}\\
\Lp_q(\A) &:= \left\{ \e\in\ZZ^m~\rm{s.t.}~A\e=\bf{0}\mod q \right\} 
\end{align*}
and translations of lattice  $\Lp_q(\A)$ defined as follows
$$\L_q^{\bf{u}}(\A) :=  \left\{ \e\in\ZZ^m~\rm{s.t.}~A\e=\bf{u}\mod q \right\}$$
for  given matrices $\A\in\ZZ^{n\times m}$ and $\bf{u}\in\ZZ_q^n$.

Let $\bf{S}=\{\s_1,\cdots,\s_k\}$ be a set of vectors in $\mathbb{R}^m$. We denote by $\|\bf{S}\|:=\max_{1 \le i \le k} \|\s_i\|$  the maximum $\ell_2$ length of the vectors in $\bf{S}$. We also denote $\tilde{\bf{S}}:=\{\tilde{\s}_1,\cdots,\tilde{\s}_k \}$ the Gram-Schmidt orthogonalization of the vectors $\s_1,\cdots,\s_k$ in that order. We refer to $\|\tilde{\bf{S}}\|$ the Gram-Schmidt norm of $\bf{S}$.

Note that for any matrix $\bf{B}\in\ZZ^{n\times m}$, there exists a singular value decomposition $\bf{B}=\bf{Q}\bf{D}\bf{P}^T$, where $\bf{Q}\in\RR^{n\times n}$, $\bf{P}\in \RR^{m\times m}$ are orthogonal matrices, and $\bf{D}\in\RR^{n\times m}$ is a diagonal matrix with nonnegative entries $s_i\ge 0$ on the diagonal, in non-increasing order. The $s_i$ are called the \textit{singular values} of $\bf{B}$. Under this convention, $\bf{D}$ is uniquely determined and $s_1(\bf{B})=\max_{\bf{u}}\|\bf{B}\bf{u}\|=\max_{\bf{u}}\|\bf{B}^T\bf{u}\|\ge \|\bf{B}\|,\|\bf{B}^T\|$ where the maxima are taken over all unit vectors $\u\in\RR^m$. Note that the singular values of $\bf{B}$ and $\bf{B}^T$ are the same.\\
    
\medskip
\noindent\textbf{Gaussian distribution.} We will use the following definitions of the discrete Gaussian distributions.

\begin{defn}
	Let $\L\subseteq\ZZ^m$ be a lattice. For a vector $\bf{c}\in\RR^m$ and a positive parameter $\sigma\in\RR$, define:
	$$\rho_{\sigma,\c}(\x)=\exp\left(- \pi\frac{\|\x-\c\|^2}{\sigma^2}\right)\quad\text{and}\quad
	\rho_{\sigma,\c}(\L)=\sum_{\x\in\L}\rho_{\sigma,\c}(\x).    $$
	The discrete Gaussian distribution over $\L$ with center $\c$ and parameter $\sigma$ is
	$$\forall \bf{y}\in\L\quad,\quad\cal{D}_{\L,\sigma,\c}(\bf{y})=\frac{\rho_{\sigma,\c}(\bf{y})}{\rho_{\sigma,\c}(\L)}.$$
\end{defn}
For convenience, we will denote by $\rho_\sigma$ and $\cal{D}_{\L,\sigma}$ for $\rho_{\sigma,\bf{0}}$ and $\cal{D}_{\L,\sigma,\bf{0}}$ respectively. When $\sigma=1$ we will write $\rho$ instead of $\rho_1$.

 It is well-known that for a vector $\bf{x}$ sampled from $\cal{D}^m_{\ZZ,\sigma}$, one has that $\|\bf{x}\|\le  \sigma\sqrt{m}$ with overwhelming probability.
\begin{lem}\label{lem:Normofsample} For $\x\gets\cal{D}_{\Lambda_q^{\bf{u}}(\A),\sigma}$, $\Pr[\|\x\|>\sigma\sqrt{m}]\le\rm{negl}(n)$.
\end{lem}


\begin{lem}\label{lem:ClosetoUniform}
For a prime $q$ and a positive integer $n$, let $m\ge n\lceil\log q\rceil$. For $\A\gets \ZZ^{n\times m}_q$, $\bf{r}\gets \cal{D}^m_{\ZZ,\sigma}$ with $\sigma\ge \omega(\sqrt{\log n})$, the distribution of $\bf{u}=\A\bf{r}\in\ZZ^{n}_q$ is statistically close to the uniform distribution over $\ZZ^{n}_q$.\\
Furthermore, fix $\bf{u}\in \ZZ^n_q$, the distribution of $\bf{r}$ conditioned on $\A\bf{r}=\bf{u}$ is $\cal{D}_{\L_q^{\bf{u}}(A),\sigma}$.
\end{lem}

The security of our construction reduces to the LWE (Learning With Errors) problem introduced by Regev~\cite{Regev05}.

\begin{defn}[Learning With Errors - LWE problem]
	Consider a prime $q$, a positive integer $n$, and a distribution $\chi$ over $\ZZ_q$. An
	$\LWE_{n,m,q,\chi}$ problem instance consists of access to an unspecified challenge oracle $\cal{O}$, being either a noisy pseudorandom sampler $\cal{O}_\s$ associated with a secret $\s\in\ZZ_q^n$, or a truly random sampler $\cal{O}_\$$ who behaviors are as follows:
	\begin{description}
		\item[$\cal{O}_\s$:] samples of the form $(\bf{a}_i,b_i)=(\bf{a}_i,\s^T\bf{a}_i+e_i)\in\ZZ_q^n\times\ZZ_q$ where $\s\in\ZZ_q^n$ is a uniform secret key, $\bf{a}_i\in\ZZ_q^n$ is uniform and $e_i\in\ZZ_q$ is a noise withdrawn from $\chi$.
		\item[$\cal{O}_\$$:] samples are uniform pairs in $\ZZ_q^n\times\ZZ_q$.
	\end{description}\
	The $\LWE_{n,m,q,\chi}$ problem allows respond queries to the challenge oracle $\cal{O}$. We say that an algorithm $\cal{A}$ decides the $\LWE_{n,m,q,\chi}$ problem if 
	$$\Adv_{\cal{A}}^{\LWE_{n,m,q,\chi}}:=\left|\Pr[\cal{A}^{\cal{O}_\s}=1] - \Pr[\cal{A}^{\cal{O}_\$}=1] \right|$$    
	is non-negligible for a random $\s\in\ZZ_q^n$.
\end{defn}

Regev~\cite{Regev05} showed that (see Theorem~\ref{thm:LWE} below) when $\chi$ is a distribution $\overline{\Psi}_\alpha$ with $\alpha\in(0,1)$, the LWE problem is hard. 

\begin{thm}\label{thm:LWE}
	If there exists an efficient, possibly quantum, algorithm for deciding the $\LWE_{n,m,q,\overline{\Psi}_\alpha}$ problem for $q>2\sqrt{n}/\alpha$ then there is an efficient quantum algorithm for approximating the $\SIVP$ and $\GapSVP$ problems, to within $\tilde{\cal{O}}(n/\alpha)$ factors in the $\ell_2$ norm, in the worst case.
\end{thm}

The theorem implies, for $n/\alpha$ is a polynomial in $n$, the LWE problem is as hard as approximating the SIVP and GapSVP problems in lattices of dimension $n$ to within polynomial (in $n$) factors.

In this paper, we will use the discrete Gaussian distribution $\cal{D}^m_{\ZZ,\sigma}$ and denote $\LWE_{n,m,q,\sigma}$ instead of $\LWE_{n,m,q,\cal{D}^{m}_{\ZZ,\sigma}}$ for convenience.

We use the following lemma which was introduced by Katsumata and Yamada in~{\cite{KY16}} to rerandomize $\LWE$ instances:
\begin{lem}\label{lem:ReRand}
Let $\ell,q,m$ be positive integers and let $r$ be a positive real number satisfying  $r\ge \max\{\omega(\sqrt{\log m}),\omega(\sqrt{\log \ell})\}$. Let $\b\in\ZZ^{m}_q$ be arbitrary and $\z\gets \cal{D}^m_{\ZZ,r}$. Then there exists an efficient algorithm $\ReRand$ such that for any $\D\in \ZZ^{m\times \ell}$ and positive real $\sigma\ge s_1(\D)$, the output of $\ReRand(\D,\b^T+\z^T,r,\sigma)$ is distributed as $\b'^T=\b^T\D+\z'^T\in\ZZ^{\ell}_q$ where the distribution of $\z'$ is close to $\cal{D}^{\ell}_{\ZZ,2r\sigma}$.
\end{lem}

\medskip
\noindent\textbf{Lattice trapdoors} 

Our work heavily bases on the notion $\G$-trapdoor introduced in~{\cite{MP12}}. In the following, we recap this notion as well as some usefull algorithms.

As in~\cite{MP12}, let $n\ge 1$, $q\ge 2$ and let $\omega=n\lceil \log q \rceil=nk$, we will use the vector $\bf{g}^T=(1,2,4,\dots,2^{k-1})$ and extend it to get the gadget matrix $\G=\bf{I}_n\otimes \bf{g}^T\in\ZZ^{n\times \omega}_q$ such that the lattice $\Lp_q(\G)$ has a public known matrix $\T_{\G}\in\ZZ^{\omega\times \omega}$ with $\Vert\widetilde{\T_{\G}}\Vert\le \sqrt{5}$ and $\|\T_{\G}\|\le \max(\sqrt{5},\sqrt{k})$.

\begin{defn}\label{Gtrapdoor}($\G$-trapdoor)
Let $n\ge 1$, $q\ge 2$ and $\omega=n\lceil \log q \rceil$, $m\ge \omega$. Let $\A\in\ZZ^{n\times m}_q$, $\G\in\ZZ^{n\times \omega}_q$. Let $\bf{H}\in\ZZ^{n\times n}_q$ be some invertible matrix. A matrix $\R\in\ZZ^{(m-\omega)\times \omega}$ is called a $\G$-trapdoor for $\A$ with tag $\bf{H}$ if it holds that $\A\begin{bmatrix} -\R \\ \bf{I}_{\omega} \end{bmatrix}=\bf{H}\G \mod q$. The quality of the trapdoor is measured by its largest singular value $s_1(\R)$.
\end{defn}

~{\cite{MP12}} also presented an algorithm to generate a pseudorandom matrix $\F\in\ZZ^{n\times (m+\omega)}_q$ together with a ``strong" $\G$-trapdoor for the lattice $\Lp_q(\F)$:
\begin{enumerate}
    \item Sample $\A\gets\ZZ^{n\times m}_q$, $\R\gets\cal{D}^{m\times \omega}_{\ZZ,\omega(\sqrt{\log n})}$ and an invertible matrix $\bf{H}\gets\ZZ^{n\times n}_q$
    \item Return $\F=[\A|\A\R+\bf{H}\G]$ and the $\G$-trapdoor $\R$.
\end{enumerate}
The matrix $\R\gets\cal{D}^{m\times \omega}_{\ZZ,\omega(\sqrt{\log n})}$ can do everything that a low-norm basis of $\Lp_q(\F)$ does. Moreover, $\R$ can be used to efficiently generate low-norm basis $\T_{\F}\in \ZZ^{(m+\omega)\times(m+\omega)}$ for $\Lp_q(\F)$. \\

Next, we recall the following lemma  from {\cite{GVP08}}:

\begin{lem}\label{SamplePre}
Let $q, k,n, m$  be integers with $q> 2, k>1$, $m>n$ and let $\A \in \ZZ_q^{n\times m}$, $\U\in\ZZ^{n\times k}_q$. Let $\T_{\A}$ be a basis for $\Lp_q(\A)$. For $\sigma\ge\|\widetilde{\T_\A}\|\cdot\omega(\sqrt{\log m})$, there is a PPT algorithm $\SamplePre(\A,\T_{\A},\U,\sigma)$ that returns a matrix $\D\in\ZZ^{n\times k}_q$ sampled from a distribution statistically close to $\cal{D}_{\L_q^{\bf{U}}(\A),\sigma}$, whenever $\L_q^{\bf{U}}(\A)$ is not empty such that $\A\D=\U$.
\end{lem}
The following lemma consists of algorithms for generating bases for lattices collected from the sampling technique in the work of Agrawal et al.~{\cite{ABB10-EuroCrypt}} and the $\SamplePre$ algorithm from the work of Micciancio et al.~{\cite[Theorem 5.1]{MP12}} which will be used in our construction. Note that the $\SamplePre$ algorithm in~\cite{MP12} is different from the $\SamplePre$ algorithm from~\cite{ABB10-EuroCrypt} in Lemma~\ref{SamplePre} above.
\begin{lem}\label{thm:Gauss}
Let $n\ge 1$, $q\ge 2$, $\omega=n\lceil \log q \rceil$, $m\ge \omega$. Let $\A\gets \ZZ^{n\times m}_q$.
    \begin{itemize} 
        \item  Let $\T_{\A}$ be a basis for $\Lp_q(\A)$, $\bf{M}\gets\ZZ_q^{n\times m_1}$ and $\sigma\ge\|\widetilde{\T_{\A}}\|\cdot\omega(\sqrt{\log(m+m_1)})$. Then there exists a PPT algorithm $\SampleBasisLeft(\A,\bf{M},\T_{\A},\sigma)$ that outputs a basis of $\Lambda_q^\perp([\A|\bf{M}])$.
        \item Let $\R\gets\cal{D}^{m\times \omega}_{\ZZ,\omega(\sqrt{\log n})}$, $\U\gets \ZZ^{n\times \omega}_q$, and let  $\bf{H}\gets \ZZ^{n\times n}_q$ be an invertible matrix. Let $\F=[\A|\A\R+\bf{H}\G]$. Then for $\sigma\ge 5 \cdot  s_1(\R) \cdot\omega(\sqrt{\log n})$, there exists a PPT algorithm $\SampleRight(\R,\F,\bf{H},\U,\sigma)$ that outputs a matrix $\bf{D}\in\ZZ^{(m+\omega)\times \omega}$ distributed statistically close to $\cal{D}_{\L_q^{\bf{U}}(\F),\sigma}$ s.t.  $\F\D=\U$.\\
        In particular, there exits a PPT algorithm $\SampleBasisRight(\R,\F,\bf{H},\U,\sigma)$ that outputs a basis $\T\in\ZZ^{(m+\omega)\times (m+\omega)}$ of $\Lambda_q^\perp(\F)$ which distributes statistically close to $\cal{D}_{\L_q^{\perp}(\F),\sigma}$, i.e., $\F\T=\bf{0}$.
    \end{itemize}
\end{lem}
Here, we note that the algorithm $\SampleBasisRight$ basically runs $\SampleRight(\R,\F,\bf{H},\bf{0},\sigma)$ many times until there are enough linearly independent output vectors to form a basis of $\Lambda_q^\perp(\F)$. According to~\cite{ABB10b}, $2(m+\omega)$ samples are needed in expectation to get the basis $\T$ for $\L_q^{\perp}(\F)$.

Peikert~\cite{Pei09} shows how to construct a basis for $\Lp_q(\A_1|\A_2|\A_3)$ from a basis for $\Lp_q(\A_2)$.
\begin{thm}\label{ExtendBasis}
    For $i=1,2,3$, let $\A_i$ be a matrix in $\ZZ^{n\times m_i}$ and let $\A=(\A_1|\A_2|\A_3)$. Let $\T_2$ be a basis of $\Lp_q(\A_2)$. There is a deterministic polynomial time algorithm $\ExtendBasis$ that outputs a basis $\T$ for $\Lp_q(\A)$ such that $\Vert \widetilde{\T}\Vert = \Vert \widetilde{\T_2}\Vert$. 
\end{thm}
We will also use the following lemma in the decryption algorithm to recover the message.

\begin{lem}\label{lem:Invert}
Let $\A$ be a uniformly random matrix in $\ZZ^{n\times m}_q$ where $m>2n$. Let $\T\in \ZZ^{m\times m}$ be a basis of $\Lp_q(\A)$. Given $\y = \s^T\A+\e^T$ where $\s\in\ZZ^n_q$, $\e\in\ZZ^m$ with $\Vert\e^T\T\Vert_{\infty}<q/4$, there exists an algorithm $\Invert(\A,\T,\y)$ that outputs $\s$ and $\e$ with overwhelming probability. 
\end{lem}

It can be easily seen that the lemma is true since the algorithm works by computing $\bf{y}^T\T \mod q = \e^T \T \mod q$. We have $\Vert\e^T\T\Vert_{\infty}<q/4$, so $\e^T \T \mod q = \e^\T \in \ZZ^{m}$. Since $\T\in \ZZ^{m\times m }$ is a basis of lattice $\Lp_q(\A)$, $\T$ has linearly independent columns, one can simply use the Gaussian elimination to recover $\e$ and then get $\s^T\A$. Finally, $\s$ can be recovered by Gaussian elimination because $\A\in \ZZ^{n\times m}$ has at least $n$ linearly independent column vectors.

\section{AHIBET Construction over Integer Lattices}

\begin{itemize}
    \item Let $\lambda$ be the security parameter, $d$ be the hierarchy depth and identities are vector $\id=(\id_1,\dots,\id_{\ell})$ ($1\le \ell\le d$) where all components $\id_i$ are in $\ZZ^n\setminus{\{\bf{0}\}}$.
    \item Let $\FRD:\ZZ^n_q \longrightarrow \ZZ^{n \times n}_q$ be a full-rank difference encoding (FRD) from~\cite{ABB10-EuroCrypt} s.t. for all distinct $\u,\v \in \ZZ^n_q$, $\FRD(\u)-\FRD(\v)\in \ZZ^{n\times n}_q$ is an invertible matrix.
    \item Let $\mathsf{H}:(\ZZ^n_q)^* \longrightarrow \ZZ^n_q$ be a collision resistant hash function.
    \item For an integer $q>2$, $x\in \ZZ_q$, the algorithm $\Round(x)$ returns $0$ if $x$ is closer to $0$ than to $\left\lfloor\displaystyle{\frac{q}{2}}\right\rfloor$ modulo $q$; otherwise, it returns $1$.
\end{itemize}
In the construction of the AHIBET scheme, we assume each identity $\id$ can only be given exactly one tracing key $\Tsk_{\id}$.

\begin{description}
	\item[Setup($\lambda,d$)]$\\$ On input security parameter $\lambda$ and a maximum hierarchy depth $d$, set the parameters $(n,m,q,\omega,\bar{\sigma},\tau,\alpha,r)$ as in section \ref{sec:params}, the algorithm does:
	\begin{enumerate}
		\item Sample uniformly random matrices $\A \gets\ZZ_q^{n\times m}$,  $ \A_2,\dots,\A_d\gets\ZZ_q^{n\times \omega}$ and $\R_0,\R_1\gets\cal{D}^{m\times \omega}_{\ZZ,\omega(\sqrt{\log n})}$.
		\item Set $\A_0 \gets \A\R_0\in\ZZ^{n\times \omega}_q$, $\A_1 \gets \A\R_1\in\ZZ^{n\times \omega}_q$ and choose $\U_1,\U_2\gets \ZZ^{n\times \lambda}_q$ uniformly at random.
		\item Output the master public key and the master secret key
		$$\MPK=(\A,\A_0,\A_1,\dots,\A_d,\U_1,\U_2)~,~\MSK=(\R_0,\R_1).$$
	\end{enumerate} 
	\item [Extract($\MPK,\MSK,\id$)]$\\$ On input the master pubic key $\MPK$, the master secret key $\MSK$ and an identity $\id$ of level $1$, the algorithm generates secret key for $\id$ as follows:
	\begin{enumerate}
	    \item Compute $\F_{\id}=[\A|\A_1+\FRD(\id)\G] \in \ZZ^{n\times (m+\omega)}_q$.
	    \item Sample $\T_{\id}\gets \SampleBasisRight(\F_{\id},\R_1,\FRD(\id),\sigma_1) \in \ZZ^{(m+\omega)\times (m+\omega)}_q$ s.t. $\F_{\id}\T_{\id}=0$.
	    \item Output $\SK_{\id}\gets \T_{\id}$.\\
	\end{enumerate}
	
	\item [Derive($\MPK,\SK_{\id},(\id|\id_{\ell})$)]$\\$ On input the master pubic key $\MPK$, a secret key $\SK_{\id}$ corresponding to an identity $\id=(\id_1,\dots,\id_{\ell-1})$ at depth $\ell-1$ and an identity $\id|\id_{\ell}=(\id_1,\dots,\id_{\ell-1},\id_{\ell})$ of level $\ell>1$, the algorithm generates secret key for $\id$ as follows:
	\begin{enumerate}
	    \item Set $\F_{\id|\id_{\ell}}=[\F_{\id}|\A_{\ell}+\FRD(\id_{\ell})\G] \in \ZZ^{n\times (m+\ell\omega)}_q$ with $\F_{\id}=[\A|\A_1+\FRD(\id_1)\G|\dots|\A_{\ell-1}+\FRD(\id_{\ell-1})\G]$.
	    \item Sample $\T_{\id|\id_{\ell}}\gets \SampleBasisLeft(\F_{\id},\A_{\ell}+\FRD(\id_{\ell})\G,\SK_{\id},\sigma_{\ell})$ s.t. $\F_{\id|\id_{\ell}}\T_{\id|\id_{\ell}}=0$.
	    \item Output $\SK_{\id|\id_{\ell}}\gets \T_{\id|\id_{\ell}}$.\\
	\end{enumerate}
	\item [TskGen($\MPK,\MSK,\id$)]$\\$ On input the master pubic key $\MPK$, the master secret key $\MSK$ and an identity $\id=(\id_1,\dots,\id_{\ell})$, the algorithm generates the tracing key for $\id$ as follows:
	\begin{enumerate}
	  
	    \item Compute $\F'_{\id}=[\A|\A_0+\FRD(\mathsf{H(\id)})\G]\in \ZZ^{n\times (m+\omega)}_q$.
	     \item Sample $\D'_{\id}\gets \SampleBasisRight(\F'_{\id},\R_0,\FRD(\mathsf{H}(\id)),\sigma_1) \in \ZZ^{(m+\omega)\times (m+\omega)}_q$ s.t. $\F'_{\id}\D'_{\id}=0$.
	    \item Sample $\D_{\id}\gets \SamplePre(\F'_{\id},\D'_{\id},\U_2 ,\sigma_{1})\in \ZZ^{(m+\omega)\times \lambda}_q$.
	    
	    \item Output $\Tsk_{\id}\gets\D_{\id}$.\\
	\end{enumerate}
	
	\item [Encrypt($\MPK,\id,\m$)]$\\$
	On input  the master pubic key $\MPK$, the algorithm encrypts the message $\m\in \{0,1\}^{\lambda}$ for identity $\id=(\id_1,\dots,\id_{\ell})$ at depth $\ell$ as follows:
	\begin{enumerate}
	    \item Compute $\F_{\id}=[\A|\B_{\id}]= [\A|\A_1+\FRD(\id_1)\G|\dots|\A_{\ell}+\FRD(\id_{\ell})\G]\in \ZZ^{n\times (m+\ell\omega)}_q$.
	    \item Sample $\k\gets \{0,1\}^{\lambda}$.
	    \item Sample a uniformly random vector $\s\gets\ZZ^n_q$.
	    \item Choose noise vectors $\e_0\gets \cal{D}^m_{\ZZ,r}$, $\e_1\gets  \cal{D}^{\ell\omega}_{\ZZ,2r\tau}$ $\e_2\gets \cal{D}^{\lambda}_{\ZZ,r},\  \e_3\gets \cal{D}^{\lambda}_{\ZZ,2r\tau}, \ \e_4\gets  \cal{D}^{\omega}_{\ZZ,2r\tau}$.
	    \item Set
	    $$ \c_0^T = \s^T \A + \e_0^T , \ \ \ \ \ \c_1^T =\s^T\B_{\id}+\e_1^T, \ \ \ \ \
	     \c_2^T = \s^T \U_1 +  \e_2^T+\m^T\left\lfloor\frac{q}{2}\right\rfloor,  $$
	    and
	    $$\c_3^T=\s^T\U_2+ \e_3^T+\k^T\left\lfloor\frac{q}{2}\right\rfloor , \ \ \ \ \  \c_4^T = \s^T (\A_0 + \FRD(\mathsf{H}(\id))\G) + \e_4^T.$$
	    \item Output $\CT=(\c_0,\c_1,\c_2,\c_3,\c_4,\k)$.\\
	\end{enumerate}
	
	\item [Decrypt($\MPK,\CT,\SK_{\id}$)]$\\$ 
		On input  the master pubic key $\MPK$, a ciphertext $\CT$ and a secret key $\SK_{\id}$ where $\id=(\id_1,\dots,\id_{\ell})$ is an identity at depth $\ell$, the algorithm does:
		\begin{enumerate}
		    \item Parse $\CT=(\c_0,\c_1,\c_2,\c_3,\c_4,\k)$; Output $\perp$ if $\CT$ doesn't parse.
		    \item Set $\F_{\id}=[\A|\A_1+\FRD(\id_1)\G|\dots|\A_{\ell}+\FRD(\id_{\ell})\G]$ and recover $\s$ via $\Invert(\SK_{\id}, \F_{\id},[\c_0^T|\c_1^T])$.
		    \item Recover $\Tilde{\k}\gets \Round(\c_3^T-\s^T\U_2)$; Return $\perp$ if $\Tilde{\k}\neq \k$.
		    \item Compute $\m\gets \Round(\c_2^T-\s^T\U_1)$.
		    \item Output $\m$.\\
		\end{enumerate}
	
	\item [TkVer($\MPK,\id,\Tsk_{\id},\CT$)]$\\$
	On input  the master pubic key $\MPK$, the algorithm uses the tracing key $\Tsk_{\id}=\D_{\id}$ corresponding to the identity $\id$ to check whether a ciphertext $\CT$ is encrypted for the given identity $\id$:
	\begin{enumerate}
	    \item Parse $\CT=(\c_0,\c_1,\c_2,\c_3,\c_4,\k)$; Output $\perp$ if $\CT$ doesn't parse.
	    \item Compute $\Tilde{\k}\gets\Round(\c_3^T-[\c_0^T|\c_4^T]\D_{\id})$.
	    \item If $\Tilde{\k}=\k$ then output $1$; else output $0$.
	\end{enumerate}
\end{description}

\subsection{Parameters}\label{sec:params}
Let $\lambda$ be the security parameter, $d$ is the maximum hierarchical depth of the scheme, $1\le \ell\le d$. We assume that all parameters are functions of $\lambda$. Now for the system to work correctly, we need to ensure:
	\begin{itemize}
		\item $\sigma_{\ell}$ is large enough for $\SampleBasisLeft$ and $\SampleBasisRight$, i.e.,
		$\sigma_{\ell}>O(m)\cdot\omega(\log n)\cdot\omega(\sqrt{\log (\ell+1)m})$,
		\item $\tau$ is large enough for $\ReRand$, i.e. 
		$\tau\ge O(m^{3/2})\cdot\omega(\log^{3/2} n),$
		\item the error term in decryption is less than ${q}/{4}$ with high probability, i.e. $\alpha<(8\sigma_{\ell}\tau(m+\ell\omega))^{-1}$,
	\end{itemize}
Hence the following choice of parameters $(n,m,q,\omega,\bar{\sigma},\tau,\alpha,r)$ satisfies all of the above conditions, taking $\lambda$ to be the security parameter:
	\begin{equation}\label{eq:params}
	\begin{aligned}
	& n\ge 2 \quad,\quad q\ge 2 \quad,\quad \omega=n\lceil\log q\rceil \quad,\quad 1\le \ell \le d\\
	& m\ge n\log q +\omega(\log n) \quad,\quad r=\alpha q \quad,\quad \sigma_1 = O(\sqrt{m})\cdot\omega(\log n)\\
	&\sigma_{\ell} = O(m)\cdot\omega(\log n)\cdot\omega(\sqrt{\log (\ell+1)m}),\\
	&\tau= O(m^{3/2})\cdot\omega(\log^{3/2} n),\\
	&\alpha=[(\ell+1)\cdot O(m^{7/2})\cdot\omega(\log^{5/2}n)\cdot\omega(\sqrt{\log (\ell+1)m})]^{-1}.
	\end{aligned}
	\end{equation}
\subsection{Correctness and soundness}\label{sec:correct}
When the cryptosystem is operated as specified, during decryption of a correctly generated ciphertext encrypted a message $\m$ to an identity $\id=(\id_1,\dots,\id_{\ell})$ at depth $\ell\le d $, with the parameters as specified in~\ref{sec:params}, we have: 
\begin{itemize}
    \item Since $\e_0\gets \cal{D}^m_{\ZZ,r}$, $\e_1\gets \cal{D}^{\ell\omega}_{\ZZ,2r\tau}$, by applying Lemma~\ref{lem:Normofsample} and the parameters set up, we get $\Vert [\e_0^T|\e_1^T]\Vert\le  2r\tau\cdot\sqrt{m+\ell\omega}$, which means
    $\Vert [\e_0^T|\e_1^T]\T_{\id}\Vert_{\infty}\le \Vert [\e_0^T|\e_1^T]\Vert \cdot \Vert\T_{\id}\Vert\le (2r\tau\cdot \sqrt{(m+\ell\omega)})\cdot(\sigma_{\ell}\cdot\sqrt{m+\ell\omega})\le 2r\tau\sigma_{\ell}\cdot(m+\ell\omega)\le q/4.$\\
    Using Lemma~\ref{lem:Invert}, it is sufficient to show that the algorithm $\Invert(\T_{\id}, \F_{\id},[\c_0^T|\c_1^T])$ will output $\s$ with overwhelming probability.\\
    \item Since $\e_2\gets \cal{D}^{\lambda}_{\ZZ,r}$, $\e_3\gets \cal{D}^{\lambda}_{\ZZ,2r\tau}$, by applying Lemma~\ref{lem:Normofsample} and the parameters set up we get $\Vert \e_2\Vert \le \Vert \e_3\Vert \le 2r\tau\sqrt{\lambda}<q/4$. \\
    Hence $\Round(\c_3^T-\s^T\U_2)=\Round\left(\k\left\lfloor\displaystyle{\frac{q}{2}}\right\rfloor^T+\e_3^T\right)$ and $\Round(\c_2-\s^T\U_1)=\Round\left(\m\left\lfloor\displaystyle{\frac{q}{2}}\right\rfloor^T+\e_2^T\right)$ will correctly recover $\k$ and $\m$.
\end{itemize}
In the algorithm $\TkVer$, $\Round(\c_3^T-[\c_0^T|\c_4^T]\D_{\id})=\Round\left(\k\left\lfloor\displaystyle{\frac{q}{2}}\right\rfloor^T+\e_3^T+[\e_0^T|\e_4^T]\D_{\id}\right)$ where $\Vert\D_{\id}\Vert\le\sigma_{1}\sqrt{m+\omega}$. Hence
by the parameters set up, $\TkVer$ will correctly recover the key $\k$.

\section{Security analysis}

\subsection{Proof of Anonymity}
In this part, we will prove that our proposed AHIBET scheme is $\Anon$ secure in the standard model.
\begin{theorem}\label{thm:ANON}
    The AHIBET scheme $\Pi := (\Setup,\Extract,\Derive,\Encrypt,\Decrypt,\TskGen,\TkVer)$ with parameters $(\lambda,n,m,q,\omega,\bar{\sigma},\tau,\alpha,r)$ as in~\eqref{eq:params} is $\Anon$ secure for the maximal hierarchy depth $d$ provided that the hardness of the $\LWE_{n,m+\lambda,q,r}$ problem holds. 
\end{theorem}

\begin{proof}
We will proceed the proof via a sequence of games where the \textbf{Game 0} is identical to the original $\Anon$ game and the adversary in the last game has advantage at most the advantage of an efficient LWE adversary.\\

Let $\mathcal{A}$ be a PPT adversary that attacks the AHIBET scheme $\Pi$ and has advantage $\Adv_{\cal{A},\Pi}^{\Anon}=\epsilon$. We will then construct a simulator $\cal{B}$ that solves the LWE problem using $\cal{A}$.\\
Let $G_i$ denote the event that the adversary $\cal{A}$ wins \textbf{Game $i$}. The adversary's advantage in \textbf{Game $i$} is $\left|\Pr[G_i]-\displaystyle{\frac{1}{2}}\right|$.

\begin{description}
	\item[\textbf{Game 0}.]~~ This is the original $\Anon$ game between the adversary $\cal{A}$ against our scheme and an $\Anon$ challenger. \\
	$$\Adv_{\cal{A},\Pi}^{\Anon}= \left|\Pr[G_0] -\frac{1}{2}\right|=\left|\Pr[b'=b]-\frac{1}{2}\right|.$$
	\item[\textbf{Game 1}.]~~ \textbf{Game 1} is analogous to \textbf{Game 0} except that we slightly modify the way that the challenger $\cal{C}$ generates the master public key $\MPK$ and responds to the key derivation oracles $\cal{O}_{\Derive}$ as well as the tracing key oracles $\cal{O}_{\TskGen}$. Let $\id^{\ast}=(\id^{\ast}_1,\dots,\id^{\ast}_{\ell})$ $(\ell\le d)$ be the target identity that $\cal{A}$ intends to attack. After receiving $\id^{\ast}$, $\cal{C}$ does:
	\begin{enumerate}
		\item Sample $\A\gets\ZZ_q^{n\times m}$, $\R_{0},\R_1,\dots,\R_d \gets\cal{D}^{m\times \omega}_{\ZZ,\omega(\sqrt{\log n})}$ and $\overline{\R}\gets \cal{D}^{m\times \lambda}_{\ZZ,\omega(\sqrt{\log n})}$.
		\item Set $\A_0 \gets \A\R_0- \FRD(\mathsf{H}(\id^{\ast}))\G$.
		\item Set $\A_i \gets \A\R_i- \FRD(\id^{\ast}_i)\G$ for $i= 1,\dots,\ell$ and $\A_i \gets \A\R_i$ for $\ell<i\le d$.
		\item Set $\U_2 \gets \A\overline{\R}$ and sample $\U_1\gets \ZZ^{n\times \lambda}_q$. 
		\item Send master public key 
	$$\MPK=(\A,\A_0,\A_1,\dots,\A_d,\U_1,\U_2)$$
	 to $\mathcal{A}$ and keep $\R_0,\R_1,\dots,\R_d, \overline{\R}$ secret.
	\end{enumerate}

	\begin{itemize}
	    \item Recall that the adversary $\cal{A}$ is not allowed to use the challenge identity $\id^{\ast}$ or its prefixes for its key derivation queries.
	    To respond to the key derivation queries $\cal{O}_{\Derive}$ for $\id=(\id_1,\dots,\id_{k})\neq \id^{\ast}=(\id^{\ast}_1,\dots,\id^{\ast}_{\ell})$ ($1\le k \le d$), $\cal{C}$ sets
	    \begin{align*}
	    \F_{\id}&=[\A|\A_1+\FRD(\id_1)\G|\dots|\A_{k}+\FRD(\id_{k})\G]
	          \in \ZZ^{n\times (m+k\omega)}_q.
		\end{align*}
	    \begin{itemize}
	    \item If $k\le \ell$, then 
	        $$\F_{\id} = [\A|\A \R_1 + (\FRD(\id_1)-\FRD(\id^{\ast}_1))\G| \ldots|\A \R_k + (\FRD(\id_k) - \FRD(\id^{\ast}_k))\G ]. $$
	    Let $h$ be the sallowest level where $\id_h\neq \id_h^{\ast}$ ($h\le k$). By the property of the full-rank difference encoding $\FRD$, $\FRD(\id_h)-\FRD(\id_h^{\ast})\in \ZZ^{n\times n}_q$ is an invertible matrix, $\cal{C}$ then samples
		$$\T_{\id_h}\gets \SampleBasisRight([\A|\A\R_h+(\FRD(\id_h)-\FRD(\id_h^{\ast}))\G],\R_h,\FRD(\id_h)-\FRD(\id_h^{\ast}),\sigma_1).$$
		If $h=k=1$, $\cal{C}$ returns $\SK_{\id} = \T_{\id_1}$.\\
		If $k>1$, $\cal{C}$ 
		uses algorithm $\ExtendBasis$ to extend the basis $\T_{\id_h}$ of $\Lp_q([\A|\A\R_h+(\FRD(\id_h)-\FRD(\id_h^{\ast}))\G])$ to a basis $\T_{\id}$ of $\Lp_q(\F_{\id})$ then returns $\SK_{\id} = \T_{\id}$.
		\item If $k>\ell$, then  
		$$\F_{\id}=[\A|\A \R_1|\ldots|\A \R_{\ell}|\A \R_{\ell+1} + \FRD(\id_{\ell+1})\G|\ldots|\A \R_k + \FRD(\id_k)\G]$$
		
		and 
		    $\FRD(\id_{\ell+1})\in \ZZ^{n\times n}_q$ is an invertible matrix. The challenger  $\cal{C}$ samples
		$$\T_{\id_{\ell+1}}\gets \SampleBasisRight([\A|\A\R_{
		\ell+1}+\FRD(\id_{\ell+1})\G],\R_{\ell+1},\FRD(\id_{\ell+1}),\sigma_1)$$
		and uses algorithm $\ExtendBasis$ to extend the basis $\T_{\id_{\ell+1}}$ of $\Lp_q([\A|\A\R_{\ell+1}+\FRD(\id_{\ell+1})\G])$ to a basis $\T_{\id}$ of $\Lp_q(\F_{\id})$. Finally, $\cal{C}$ returns $\SK_{\id} = \T_{\id}$.
	    \end{itemize}

	    \item 	To respond to the tracing key query $\cal{O}_{\TskGen}$ for $\id=(\id_1,\dots,\id_{k})\neq \id^{\ast}=(\id^{\ast}_1,\dots,\id^{\ast}_{\ell})$, $\cal{C}$ sets
		$$\F'_{\id}=[\A|\A_0 + \FRD(\mathsf{H}_{\id})\G]=[\A|\A\R_{0}+(\FRD(\mathsf{H}(\id))- \FRD(\mathsf{H}(\id^{\ast})))\G].$$
		Since $\mathsf{H}$ is a collision resistant hash function, $\mathsf{H}(\id)\neq \mathsf{H}(\id^{\ast})$ even if $\id$ is a prefix of $\id^{\ast}$ and thus $\FRD(\mathsf{H}(\id))- \FRD(\mathsf{H}(\id^{\ast}))$ is an invertible matrix in $\ZZ^{n\times n}_q$. Challenger $\cal{C}$ samples
		$$\D'_{\id}\gets \SampleBasisRight(\F'_{\id},\R_{0},\FRD(\mathsf{H}(\id))-\FRD(\mathsf{H}(\id^{\ast}))\G,\sigma_1)$$
		then invokes the algorithm $\SamplePre$
		$$\D_{\id}\gets \SamplePre(\F'_{\id},\D'_{\id},\U_2 ,\sigma_{1})\in \ZZ^{(m+\omega)\times \lambda}_q$$
		and returns $\Tsk_{\id}=\D_{\id}$.
		
	    
	\end{itemize}
	Using Lemma~\ref{lem:ClosetoUniform}, we can easily prove that the matrices $\A_i$ ($0\le i\le d$) are statistically close to uniform. Hence, in the adversary's point of view, $\A_0,\A_1,\dots,\A_d$ in \textbf{Game 0} and \textbf{Game 1} are computationally indistinguishable.
	
    Next, we consider the responses to the secret key derivation queries $\cal{O}_{\Derive}$ and the tracing key queries $\cal{O}_{\TskGen}$. For secret key derivation queries $\cal{O}_{\Derive}$, Theorem~\ref{thm:Gauss} shows that when $\sigma_1\ge5\cdot s_1(\R)\cdot\omega(\sqrt{\log n})$, $\sigma_{\ell}\ge\|\widetilde{\T_{\id_h}}\|\cdot\omega(\sqrt{\log(m+{\ell}\omega)})$, the algorithms  $\SampleBasisRight$ and $\ExtendBasis$ generate a basis $\T_{\id}$ for $\Lp_q(\F_{\id})$ which is statistically close to the one generated in the original game. Similarly, the tracing keys generate by $\SampleBasisRight$ and $\SamplePre$ in \textbf{Game 1} have distribution statistically close to ones  in \textbf{Game 0}. 
	
	Since the master public key $\MPK$ and responses to key derivation queries and tracing key queries in \textbf{Game 1} are statistically close to those in \textbf{Game 0}, these games are statistically indistinguishable in the view of the adversary. Thus we have
	$$|\Pr[G_1]-\Pr[G_0]|\le \negl(\lambda).$$

	\item[\textbf{Game 2}.]~~
    In this game, we change the way the challenge ciphertext $\CT^*$ for the challenge identity $\id^{\ast}$ is created. Recall that, after receiving a message $\m\in\{0,1\}^{\lambda}$ from the adversary $\cal{A}$, the challenger $\cal{C}$ then selects a random bit $b \in\{0,1\}$.
 
    If $b =1$, $\cal{C}$ chooses a random identity $\id'$ in the identity space which is not identical to any query identities in \textbf{Phase 1}. $\cal{C}$ then runs $\Encrypt(\MPK,\id,\m$) and sends the resulting ciphertext $\CT_1^*$ to $\cal{A}$.\\   
    If $b=0$, the challenger $\cal{C}$ does the following steps to generate $\CT_0^{\ast}$ and sends it to $\cal{A}$. 
    \begin{enumerate}
	    \item Sample $\k\gets \{0,1\}^{\lambda}$.
	    \item Sample $\s\gets\ZZ^n_q$.
	    \item Choose noise vectors $\e_0\gets \cal{D}^m_{\ZZ,r}$, $\e_2\gets \cal{D}^{\lambda}_{\ZZ,r}$.
        \item Set $\c_0^T=\s^T\A+\e_0^T$, $\c_2^T=\s^T\U_1+\e_2^T+\m^T\left\lfloor\frac{q}{2}\right\rfloor$ and
        \begin{align*}
	    \c_1^T &\gets\ReRand (\R,\c_0^T,r,\tau)\\
		\c_3^T &\gets  \ReRand(\overline{\R},\c_0^T,r,\tau) +\left(\k\left\lfloor\frac{q}{2}\right\rfloor\right)^T\\
		\c_4^T & \gets \ReRand (\R_{0},\c_0^T,r,\tau)
		\end{align*}
        where $\R=[\R_1|\dots|\R_{\ell}]$.
	    \item Output $\CT_0^{\ast}=(\c_0,\c_1,\c_2,\c_3,\c_4,\k)$.
    \end{enumerate}
    Observe that $\c_0$ and $\c_3$ are distributed exactly as they as in the previous game. Since
     \begin{align*}
	    \F_{\id^{\ast}} &=[\A|\A_1+\FRD(\id^{\ast}_1)\G|\dots|\A_{\ell}+\FRD(\id^{\ast}_{\ell})\G]\in \ZZ^{n\times (m+\ell\omega)}_q\\
		&=[\A|\A\R_1|\dots|\A\R_{\ell}] = [\A|\A\R]
		\end{align*}
	by  Lemma~\ref{lem:ReRand}, we get
	\begin{align*}
	   \c_1^T &=\ReRand (\R=[\R_1|\dots|\R_{\ell}],\c_0^T =\s^T\A +\e_0^T,r,\tau)
		=\s^T\A\R +\e^T_1, \\
	\c_3^T &=\ReRand (\overline{\R},\c_0^T =\s^T\A +\e_0^T,r,\tau)+\left(\k\left\lfloor\frac{q}{2}\right\rfloor\right)^T\\
		&=\s^T\A\overline{\R} +\e^T_3 =\s^T\U_2 +\e^T_3+\k^T\left\lfloor\frac{q}{2}\right\rfloor, \\
	\c_4^T &=\ReRand (\R_{0},\c_0^T =\s^T\A +\e_0^T,r,\tau)\\
		&=\s^T\A\R_{0} +\e^T_4 =\s^T(\A_0+\FRD(\mathsf{H}(\id^{\ast}))\G) +\e_4^T
	\end{align*}
	where the distribution of $\e_1,\e_3$ and $\e_4$ are statistically close to $\cal{D}^{\ell\omega}_{\ZZ,2r\tau}$, $\cal{D}^{\lambda}_{\ZZ,2r\tau}$ and $\cal{D}^{\omega}_{\ZZ,2r\tau}$, respectively. So we yields that \textbf{Game 1} and \textbf{Game 2} are statistically close in the adversary's point of view, the adversary's advantage against \textbf{Game 2} will be the same as \textbf{Game 1}.
	$$|\Pr[G_2]-\Pr[G_1]|\le \negl(\lambda).$$
	
    
\end{description}
Theorem~\ref{thm:ANON} then follows from the reduction from the LWE problem by the following lemma. 

\begin{lem}\label{lem:LWEreduction} If there exists an adversary $\cal{A}$ that wins the Game 2 with non-negligible advantage  then there is an adversary $\cal{B}$ that solves the $\LWE$ problem, i.e.,  $\Adv_{\cal{A}}^{Game 2} \le \Adv_{\cal{B}}^{\LWE_{n,m+\lambda,q,r}}(\lambda) $ for some $\LWE$ adversary $\cal{B}$.
\end{lem}

\textit{Proof of Lemma~\ref{lem:LWEreduction}.} Recall that an $\LWE$ problem instance is provided as a sampling oracle $\cal{O}$. 
$\cal{B}$ requests from oracle $\cal{O}$ and receives a  decisional $\LWE_{n,m+\lambda,q,r}$ problem sample $(\C,\c^T = \u^T+\e^T)$ where $\C$ is a random matrix in $\ZZ^{n\times (m+\lambda)}_q$, $\c\in\ZZ^{m+\lambda}$ and $\e$ is sampled from the distribution $\cal{D}^{m+\lambda}_{\ZZ,r}$. $\cal{B}$ needs to decide whether $\u$ is truly random $\cal{O}_\$$ or a noisy pseudo-random $\cal{O}_s$ for some secret random $\bf{s}\in \ZZ^n_q$ such that $\u^T = \s^T\C$. $\cal{B}$ simulates \textbf{Game 2} with adversary $\cal{A}$ and uses the guess from $\cal{A}$ to respond LWE challenges.\\

After receiving the challenge identity $\id^{\ast}=(\id^{\ast}_1,\dots,\id^{\ast}_{\ell})$ ($\ell\le d$) from $\cal{A}$, $\cal{B}$ constructs the simulator as follows:

\begin{itemize}
\item Split $\C=[\A|\U_1]$ where $\A\in\ZZ^{n\times m}_q$ and $\U_1\in\ZZ^{n\times \lambda}_q$.
\item Sample $\R_0,\R_1,\dots,\R_d \gets\cal{D}^{m\times \omega}_{\ZZ,\omega(\sqrt{\log n})}$, $\overline{\R}\gets \cal{D}^{m\times \lambda}_{\ZZ,\omega(\sqrt{\log n})}$ and set $\R=[\R_1|\dots|\R_{\ell}]$.
\item Set $\A_0\gets \A\R_0- \FRD(\mathsf{H}(\id^{\ast}))\G$.
\item Set $\A_i \gets \A\R_i- \mathsf{H}(\id_i^{\ast})\G$ for $i= 1,\dots,\ell$ and $\A_i \gets \A\R_i$ for $\ell<i\le d$.
\item Set $\U_2 \gets \A\overline{\R}$ and sample $\U_1\gets \ZZ^{n\times \lambda}_q$. 
		\item Send the master public key 
	$$\MPK=(\A,\A_0,\A_1,\dots,\A_d,\U_1,\U_2)$$
	 to $\mathcal{A}$ and keep $\R_0,\R_1,\dots,\R_d, \overline{\R}$ secret.
\item Respond to the key derivation queries and tracing key queries as in Game 2.
\item Split $\c^T = (\bar{\c}^T|\Tilde{\c}^T) \in \ZZ_q^{m+\lambda}$ where $\Bar{\c}^T=\Bar{\u}^T + \Bar{\e}^T \in \ZZ_q^m $ and $\Tilde{\c}^T=\Tilde{\u}^T+\Tilde{\e}^T\in\ZZ^{\lambda}_q$.
\item Create the challenge ciphertext $\CT^{\ast}$:
    \begin{enumerate}
	    \item Sample $\k\gets \{0,1\}^{\lambda}$.
        \item Set $\c_0^T=\Bar{\c}^T$, $\c_2^T=\Tilde{\c}^T+\m^T\left\lfloor\displaystyle{\frac{q}{2}}\right\rfloor$. 
        \item Set
       \begin{align*}
	    \c_1^T &\gets\ReRand (\R,\c_0^T,r,\tau)\\
		\c_3^T &\gets  \ReRand(\overline{\R},\c_0^T,r,\tau) +\k^T\left\lfloor\frac{q}{2}\right\rfloor\\
		\c_4^T & \gets \ReRand (\R_{0},\c_0^T,r,\tau)
		\end{align*}
        \item Send  $\CT^{\ast}=(\c_0,\c_1,\c_2,\c_3,\c_4,\k)$ to $\cal{A}$.
    \end{enumerate}
\end{itemize}

When $\LWE$ oracle is pseudorandom (i.e. $\cal{O}=\cal{O}_s)$ then $\c_0^T=\s^T\A+\e_0^T$, $\c_2^T=\s^T\U_1+\e^T_2+\m^T\left\lfloor\displaystyle{\frac{q}{2}}\right\rfloor$, meaning that $\CT^{\ast}$ is a valid challenge ciphertext that encrypts challenge message $\m$ for the target identity $\id^*$.\\
When $\LWE$ oracle is a random oracle (i.e. $\cal{O}=\cal{O}_\$$), $\c^T$ is uniformly random in $\ZZ^{m+\lambda}_q$ and thus $\CT^{\ast}$ distributes as a ciphertext encrypted for a random identity in the identity space.
Indeed, we have
    \begin{align*}
	    \c_0^T &=\Bar{\u}^T+\e_0^T\\
	    \c_1^T &=\ReRand (\R,\c_0^T,r,\tau)=\Bar{\u}^T\R+\e_1^T\\
	    \c_2^T &=\Tilde{\u}^T+\e_2^T+\m^T\left\lfloor\displaystyle{\frac{q}{2}}\right\rfloor \\
	    \c_3^T&=\ReRand(\overline{\R},\c_0^T,r,\tau) +\m^T\left\lfloor\frac{q}{2}\right\rfloor=\bar{\u}^T\overline{\R}+\e_3^T +\k^T\left\lfloor\frac{q}{2}\right\rfloor\\
        \c_4^T&=\ReRand (\R_{0},\c_0^T,r,\tau)=\Bar{\u}^T\R_0+\e_4^T
    \end{align*}
where $\e_1$, $\e_3$ and $\e_4$ are statistically close to $\cal{D}^{\ell\omega}_{\ZZ,2r\tau}$, $\cal{D}^{\lambda}_{\ZZ,2r\tau}$ and $\cal{D}^{\omega}_{\ZZ,2r\tau}$, respectively. Since ${\u}=(\bar{\u}|\Tilde{\u})$ is a random vector, the following distributions are negligibly close by using Lemma~\ref{lem:ClosetoUniform}:
$$(\A,\A\R,\A\overline{\R},\A\R_{0},\Bar{\u}^T,\Bar{\u}^T\R,\Bar{\u}^T\overline{\R},\Bar{\u}^T\R_{0})
\approx(\A,\B_{\id'},\U_2,\A_0+\FRD(\mathsf{H}(\id^{\ast}))\G,\Bar{\u}^T,\u_1^T,\u_2^T,\u_3^T)$$ 
where $\A$ is a random matrix in $\ZZ^{n\times m}_q$, $\R_0,\R_1,\dots,\R_d \gets\cal{D}^{m\times \omega}_{\ZZ,\omega(\sqrt{\log n})}$, $\overline{\R}\gets \cal{D}^{m\times \lambda}_{\ZZ,\omega(\sqrt{\log n})}$ , $\R=[\R_1|\dots|\R_{\ell}]$ ($\ell\le d$), $\B_{\id'} = [\A_1+\FRD(\id'_1)\G|\dots|\A_{\ell}+\FRD(\id'_{\ell})\G]$ for a random identity $ \id'=(\id'_1,\ldots,\id'_{\ell})$ of level $\ell$, and $\u_1\in\ZZ^{\ell\omega}_q$, $\u_2\in\ZZ^{\lambda}_q$, $\u_3\in\ZZ^{\omega}_q$ are uniformly random vectors. Therefore, in the view of the adversary $\cal{A}$, when the LWE oracle is random, $\CT^{\ast}$ distributes as a ciphertext encrypted message $\m$ for a random identity. This implies that
    $$\Adv_{\cal{A}}^{Game 2} \le \Adv_{\cal{B}}^{\LWE_{n,q,m+\lambda,r}}(\lambda). $$
\end{proof}

\subsection{Proof of Ciphertext Indistinguishability}
Finally, we will prove that our proposed AHIBET scheme is $\IND$ secure in the standard model. Recall that indistinguishable from random meaning that the challenge ciphertext encrypted for a given message $\m^{\ast}$ is computationally indistinguishable from a challenge ciphertext encrypted for a random message $\m$ in the message space on the same challenge identity $\id^{\ast}$.
\begin{theorem}\label{thm:IND}
    The AHIBET scheme $\Pi := (\Setup,\Extract,\Derive,\Encrypt,\Decrypt,\TskGen,\TkVer)$ with parameters $(\lambda,n,m,q,\omega,\bar{\sigma},\tau,\alpha,r)$ as in~\eqref{eq:params}  is $\IND$ secure for the maximal hierarchy depth $d$ provided that the hardness $\LWE_{n,q,m+\lambda,r}$ assumption holds.
\end{theorem}

\begin{proof}
We will proceed the proof via a sequence of games where the \textbf{Game 0} is identical to the original $\IND$ game and the adversary has no advantage in winning the last game. Let $\mathcal{A}$ be a PPT adversary that attacks the AHIBET scheme $\Pi$ and has advantage $\Adv_{\cal{A},\Pi}^{\IND}=\epsilon$. We will then construct a simulator $\cal{B}$ that solves the LWE problem using $\cal{A}$.\\
In \textbf{Game $i$}, let $G_i$ denote the event that the adversary $\cal{A}$ win the game. The adversary's advantage in \textbf{Game $i$} is $|\Pr[G_i] - \frac{1}{2} |$.

\begin{description}
	\item[\textbf{Game 0}.]~~ This is the original $\IND$ game between the adversary $\cal{A}$ against our scheme and an $\IND$ challenger. 
	$$\Adv_{\cal{A},\Pi}^{\IND}= \left|\Pr[G_0] -\frac{1}{2}\right|=\left|\Pr[b'=b]-\frac{1}{2}\right|.$$

	\item[\textbf{Game 1}.]~~ \textbf{Game 1} is similar to \textbf{Game 0} except that we slightly modify the way that the challenger $\cal{C}$ generates the master public key $\MPK$ and responds to the key derivation oracles $\cal{O}_{\Derive}$ and tracing key oracles $\cal{O}_{\TskGen}$. Let $\id^{\ast}=(\id^{\ast}_1,\dots,\id^{\ast}_{\ell})$ $(\ell\le d)$ be the identity that $\cal{A}$ intends to attack. After receiving $\id^{\ast}$, $\cal{C}$ does:
	\begin{enumerate}
		\item Sample $\A\gets\ZZ_q^{n\times m}$, $\R_0,\R_1,\dots,\R_d \gets\cal{D}^{m\times \omega}_{\ZZ,\omega(\sqrt{\log n})}$ and $\D_0\gets\cal{D}^{m\times\lambda}_{\ZZ,\sigma_{1}}$, $\D_1\gets \cal{D}^{\omega\times\lambda}_{\ZZ,\sigma_{1}}$.
		\item Set $\A_i \gets \A\R_i- \FRD(\id^{\ast})\G$ for $i= 1,\dots,\ell$ and $\A_i \gets \A\R_i$ for $\ell<i\le d$.
		\item Set $\overline{\R} \gets \D_0+\R_0\D_1$.
		\item Set $\A_0\gets \A\R_{0}- \FRD(\mathsf{H}(\id^{\ast}))\G$, $\U_2 \gets \A \overline{\R}$ and sample $\U_1\gets \ZZ^{n\times \lambda}_q$. 
		\item Output the master public key
	$$\MPK=(\A,\A_0,\A_1,\dots,\A_d,\U_1,\U_2)$$
		and keep $\R_0,\R_1,\dots,\R_d,\overline{\R}$ secret.
	\end{enumerate}
	\begin{itemize}
	    \item 	The adversary $\cal{A}$ is not allowed to ask for the key derivation queries of the challenge identity $\id^{\ast}$ and its prefixes.
	    To respond to a key derivation query $\cal{O}_{\Derive}$ for an identity  $\id=(\id_1,\dots,\id_{k})$, $\cal{C}$ sets:
	     \begin{align*}
	    \F_{\id}=[\A|\A_1+\FRD(\id_1)\G|\dots|\A_{k}+\FRD(\id_{k})\G]\in \ZZ^{n\times (m+k\omega)}_q.
		\end{align*}
        \begin{itemize}
	    \item If $k\le \ell$, then 
	   $$\F_{\id} = [\A|\A \R_1 + (\FRD(\id_1)-\FRD(\id^{\ast}_1))\G| \ldots|\A \R_k + (\FRD(\id_k) - \FRD(\id^{\ast}_k))\G ].$$
	    
	    Let $h$ be the sallowest level where $\id_h\neq \id_h^{\ast}$ ($h\le k$). By the property of the full-rank difference encoding $\FRD$, $\FRD(\id_h)-\FRD(\id_h^{\ast})\in \ZZ^{n\times n}_q$ is an invertible matrix, $\cal{C}$ then samples
		$$\T_{\id_h}\gets \SampleBasisRight([\A|\A\R_h+(\FRD(\id_h)-\FRD(\id_h^{\ast}))\G],\R_h,\FRD(\id_h)-\FRD(\id_h^{\ast}),\sigma_1).$$
		If $h=k=1$, $\cal{C}$ returns $\SK_{\id} = \T_{\id_1}$.\\
		If $k>1$, $\cal{C}$ 
		uses algorithm $\ExtendBasis$ to extend the basis $\T_{\id_h}$ of $\Lp_q([\A|\A\R_h+(\FRD(\id_h)-\FRD(\id_h^{\ast}))\G])$ to a basis $\T_{\id}$ of $\Lp_q(\F_{\id})$ then returns $\SK_{\id} = \T_{\id}$.
		\item If $k>\ell$, then  
		$$\F_{\id}=[\A|\A \R_1|\ldots|\A \R_{\ell}|\A \R_{\ell+1} + \FRD(\id_{\ell+1})\G|\ldots|\A \R_k + \FRD(\id_k)\G] $$
		and $\FRD(\id_{\ell+1})\in \ZZ^{n\times n}_q$ is an invertible matrix. The challenger  $\cal{C}$ samples
		$$\T_{\id_{\ell+1}}\gets \SampleBasisRight([\A|\A\R_{
		\ell+1}+\FRD(\id_{\ell+1})\G],\R_{\ell+1},\FRD(\id_{\ell+1}),\sigma_1)$$
		and uses algorithm $\ExtendBasis$ to extend the basis $\T_{\id_{\ell+1}}$ of $\Lp_q([\A|\A\R_{\ell+1}+\FRD(\id_{\ell+1})\G])$ to a basis $\T_{\id}$ of $\Lp_q(\F_{\id})$. Finally, $\cal{C}$ returns $\SK_{\id} = \T_{\id}$.
	    \end{itemize}
	    \item 	To respond to the tracing key query $\cal{O}_{\TskGen}$ for $\id=(\id_1,\dots,\id_{k})\neq \id^{\ast}=(\id^{\ast}_1,\dots,\id^{\ast}_{\ell})$, $\cal{C}$ sets:
	   
		$$\F'_{\id}=[\A|\A\R_{0}+(\FRD(\mathsf{H}(\id))- \FRD(\mathsf{H}(\id^{\ast})))\G]$$
		Since $\mathsf{H}$ is a collision resistant hash function, $\mathsf{H}(\id)\neq \mathsf{H}(\id^{\ast})$ even if $\id$ is a prefix of $\id^{\ast}$ and thus $\FRD(\mathsf{H}(\id))- \FRD(\mathsf{H}(\id^{\ast}))$ is an invertible matrix in $\ZZ^{n\times n}_q$. The challenger $\cal{C}$ samples
		$$\D'_{\id}\gets \SampleBasisRight(\F'_{\id},\R_{0},\FRD(\mathsf{H}(\id))-\FRD(\mathsf{H}(\id^{\ast}))\G,\sigma_1)$$
		then invokes the algorithm $\SamplePre$
		$$\D_{\id}\gets \SamplePre(\F'_{\id},\D'_{\id},\U_2 ,\sigma_{1})\in \ZZ^{(m+\omega)\times \lambda}_q$$
		and returns $\Tsk_{\id}=\D_{\id}$.
	    
	    To respond to the tracing key query of $\id^{\ast}=(\id^{\ast}_1,\dots,\id^{\ast}_{\ell})$, the challenger $\cal{C}$ sets:\\
	    $$\Tsk_{\id^{\ast}}=\D_{\id^{\ast}}=\begin{bmatrix}
        \D_0 \\\D_1 
        \end{bmatrix}\in \ZZ^{(m+\omega)\times \lambda}$$
        so that $$[\A|\A_0+\FRD(\mathsf{H}(\id^{\ast}))\G]\D_{\id^{\ast}} = [\A|\A \R_0] \begin{bmatrix}
        \D_0 \\\D_1 
        \end{bmatrix} = \A \D_0 + \A \R_0 \D_1 = \A \overline{\R} = \U_2.$$
        Using Lemma~\ref{lem:ClosetoUniform} with $\sigma_{1}\ge O(\sqrt{m})\cdot \omega(\log n)$ we have: the distribution of $\U_2$ is statistically close to uniform over $\ZZ^{\lambda}_q$ and $\D_{\id^{\ast}}$ has the distribution $\cal{D}^{(m+\omega)\times \lambda }_{\ZZ,\sigma_{1}}$. 
        Since the master public key $\MPK$ and responses to key derivation queries and the tracing key queries are statistically close to those in \textbf{Game 0}, the adversary's advantage in \textbf{Game 1} is at most negligibly different form its advantage in \textbf{Game 0}.
    	$$|\Pr[\G_1]-\Pr[\G_0]|\le \negl(\lambda).$$
	\end{itemize}

	\item[\textbf{Game 2}.]~~ \textbf{Game 2} is similar to \textbf{Game 1} except that we modify the construction of the challenge ciphertext $\CT^*$. 
	The challenger $\cal{C}$ sets $\m_0=\m^{\ast}$, chooses a random bit $b\in\{0,1\}$, a random message $\m_1$ in the message space and generates the ciphertext $\CT_b^{\ast}$ for a message $\m_b\in\{0,1\}^{\lambda}$ of the identity $\id^{\ast}$ as follows:
    \begin{enumerate}
	    \item Sample $\k\gets \{0,1\}^{\lambda}$.
	    \item Sample $\s\gets\ZZ^n_q$.
	    \item Choose noise vectors $\e_0\gets \cal{D}^m_{\ZZ,r}$, $\e_2\gets \cal{D}^{\lambda}_{\ZZ,r}$.
        \item Set $\c_0^T=\s^T\A+\e_0^T$, $\c_2^T=\s^T\U_1+\e_2^T+\m^T_b\left\lfloor\displaystyle{\frac{q}{2}}\right\rfloor$ and
        \begin{align*}
	    \c_1^T &\gets\ReRand (\R,\c_0^T,r,\tau)\\
		\c_3^T &\gets  \ReRand(\overline{\R},\c_0^T,r,\tau) +\left(\k\left\lfloor\frac{q}{2}\right\rfloor\right)^T\\
		\c_4^T & \gets \ReRand (\R_{0},\c_0^T,r,\tau)
		\end{align*}
        where $\R=[\R_1|\dots|\R_{\ell}]$ and $\overline{\R}= \D_0+\R_0\D_1$. Note that $\|\overline{\R}\|<\tau$ by the way that the game generates the matrices $\R_0,\D_0$ and $\D_1$.
	    \item Output $\CT_0^{\ast}=(\c_0,\c_1,\c_2,\c_3,\c_4,\k)$.
    \end{enumerate}
    We have that $\c_0$, $\c_2$ are distributed exactly as in the previous game, besides
	\begin{align*}
	   \c_1^T &=\ReRand (\R=[\R_1|\dots|\R_{\ell}],\c_0^T =\s^T\A +\e_0^T,r,\tau)\\
		&=\s^T\A\R +\e^T_1 \\
	\c_3^T &=\ReRand (\overline{\R},\c_0^T =\s^T\A +\e_0^T,r,\tau)+\k^T\left\lfloor\frac{q}{2}\right\rfloor=\s^T\A\overline{\R} +\e^T_2 +\k^T\left\lfloor\frac{q}{2}\right\rfloor\\
		&=\s^T\A(\D_0+\R_0\D_1) +\e^T_3+\k^T\left\lfloor\frac{q}{2}\right\rfloor=\s^T\U_2 +\e_3^T+\k^T\left\lfloor\frac{q}{2}\right\rfloor \\
	\c_4^T &=\ReRand (\R_0,\c_0^T =\s^T\A +\e_0^T,r,\tau)\\
		&=\s^T\A\R_0 +\e^T_4 =\s^T(\A_0+\FRD(\mathsf{H}(\id^{\ast})\G) +\e_4^T
	\end{align*}
	where the distribution of $\e_1$, $\e_3$ and $\e_4$ are statistically close to $\cal{D}^{\ell\omega}_{\ZZ,2r\tau}$, $\cal{D}^{\lambda}_{\ZZ,2r\tau}$ and $\cal{D}^{\omega}_{\ZZ,2r\tau}$, respectively. So we yields that \textbf{Game 1} and \textbf{Game 2} are statistically close in the adversary's point of view, the adversary's advantage against \textbf{Game 2} will be the same as \textbf{Game 1}.
	$$|\Pr[G_2]-\Pr[G_1]|\le \negl(\lambda).$$

	\item[\textbf{Game 3}.]~~ In this game, we keep changing how the challenge ciphertext is created. The challenger $\cal{C}$ does:
    \begin{enumerate}
	    \item Sample $\k\gets \{0,1\}^{\lambda}$.
	    \item Sample $\Bar{\u}\gets\ZZ^m_q$ and $\widetilde{\u}\gets\ZZ^{\lambda}_q$
	     \item Choose noise vectors $\e_0\gets \cal{D}^m_{\ZZ,r}$, $\e_2\gets \cal{D}^{\lambda}_{\ZZ,r}$.
        \item Set $\c_0 =\Bar{\u}+\e_0$, $\c_2=\widetilde{\u}+\e_2+\m_b\left\lfloor\frac{q}{2}\right\rfloor$ and
        \begin{align*}
	    \c_1^T &\gets\ReRand (\R,\c_0^T,r,\tau)\\
		\c_3^T &\gets  \ReRand(\overline{\R},\c_0^T,r,\tau) +\left(\k\left\lfloor\frac{q}{2}\right\rfloor\right)^T\\
		\c_4^T & \gets \ReRand (\R_0,\c_0^T,r,\tau)
		\end{align*}
        where $\R=[\R_1|\dots|\R_{\ell}]$.
	    \item Output $\CT_0^{\ast}=(\c_0,\c_1,\c_2,\c_3,\c_4,\k)$.
    \end{enumerate}

Observe that the ciphertext $\c_2=\widetilde{\u}+\e_2+\m_b\left\lfloor\displaystyle{\frac{q}{2}}\right\rfloor$ in \textbf{Game 3} is uniformly random over $\ZZ^{\lambda}_q$. Therefore, the ciphertext is independent from $\m_b$ in the adversary $\cal{A}$'s view. Hence, both $\CT_0^{\ast}$ and $\CT_1^{\ast}$ is statistically close to the uniform distribution over the ciphertext space, and the adversary $\cal{A}$ has no advantage in winning the game. We have
$$\left|\Pr[G_3]-\displaystyle{\frac{1}{2}}\right|\le \negl{(\lambda)}.$$

Moreover, using the same reduction technique as in the Anonymity Game in the previous subsection, we can construct a simulator $\cal{B}$ that solves $\LWE$ problem if adversary $\cal{A}$ is able to distinguish between \textbf{Game 2} and \textbf{Game 3}. Therefore we have
$$|\Pr[G_3]-\Pr[G_2]|\le \Adv^{\LWE_{n,q,m+\lambda,r}}_{\cal{B}}(\lambda),$$
which completes the proof of Theorem 2.

\end{description}

\end{proof}

\section{Conclusion}
In this paper, we propose a Lattice-based Anonymous Hierarchical Identity-Based Encryption scheme with Traceable Identities (AHIBET) and prove that our scheme is secure in the standard model based on the decisional LWE assumption.

\bibliographystyle{alpha}
\bibliography{latbib}

\end{document}